\documentclass[10pt,conference]{styles/IEEEtran}
\IEEEoverridecommandlockouts                              % This command is only
\usepackage{tikz}
\usetikzlibrary{positioning} % help win positioning for block diagrams
\usepackage[tbtags]{amsmath}
\usepackage{enumerate}
\usepackage{graphicx}
\usepackage{amsfonts}
\usepackage{caption}
\usepackage{subcaption}
\DeclareMathOperator*{\argmin}{arg\,min}
\DeclareMathOperator{\diag}{diag}
\DeclareMathOperator{\blkdiag}{blkdiag}
\DeclareMathOperator{\col}{col}

\DeclareMathOperator{\real}{Re}

\DeclareMathOperator{\interior}{int}
\DeclareMathOperator{\st}{s.t.}

\newtheorem{thm}{Theorem}
\newtheorem{lemma}{Lemma}

\newtheorem{defn}{Definition}

\newtheorem{remark}{Remark}
\newtheorem{asmp}{Assumption}

\newcommand{\reals}{\mathbb{R}}

\tikzstyle{block} = [draw, fill=white, rectangle, 
    minimum height=2em, minimum width=2em]
\tikzstyle{sum} = [draw, fill=white, circle, node distance=0.5cm, inner sep=0pt, minimum size=0.25cm]
\tikzstyle{input} = [coordinate]
\tikzstyle{output} = [coordinate]
\tikzstyle{pinstyle} = [pin edge={to-,thin,black}]
\tikzstyle{branch}=[fill,shape=circle,minimum size=3pt,inner sep=0pt]
\tikzstyle{gnode} = [draw, fill=white, circle, node distance=0.5cm, inner sep=0pt, minimum size=0.5cm]

\begin{document}

\title{\Large{An inexact-penalty method for GNE seeking in games with dynamic agents}} % Games


\author{Andrew R.~Romano and Lacra Pavel% 
%\thanks{This work was supported by NSERC. }% (261764). %}% <-this % stops a space
	\thanks{
	This work was supported by NSERC Alliance and Huawei. A. R. Romano and L. Pavel are with ECE Department, University of Toronto, Canada. % Toronto, ON, M5S 3G4,  %Emails:
	{\tt\small andrew.romano@mail.utoronto.ca, pavel@ece.utoronto.ca}}%
}

\maketitle

\begin{abstract}
We consider a network of autonomous agents whose outputs are actions in a game with coupled constraints. In such network scenarios, agents seeking to minimize coupled cost functions using distributed information while satisfying the coupled constraints. Current methods consider the small class of multi-integrator agents using primal-dual methods. These methods can only ensure constraint satisfaction in steady-state. In contrast, we propose an inexact penalty method using a barrier function for nonlinear agents with equilibrium-independent passive dynamics. We show that these dynamics converge to an $\varepsilon$-GNE while satisfying the constraints for all time, not only in steady-state. We develop these dynamics in both the full-information and partial-information settings. In the partial-information setting, dynamic estimates of the others' actions are used to make decisions and are updated through local communication. Applications to optical networks and velocity synchronization of flexible robots are provided.
\end{abstract}

\section{Introduction}
Game theory has become a widely used tool in the control of multi-agent systems, having found many areas of applications such as in power control of communication networks \cite{ab05} and formation control for robotics \cite{lqs14}. In a generalized game, the relevant equilibrium is the generalized Nash equilibrium, or GNE. At a GNE, each agent (player) is minimizing its own, coupled cost-function subject to the constraints, given that the other agents' actions remain fixed. Many GNE seeking algorithms assume that each agent has full-information of the others' actions and/or each has has no inherent dynamics. These are restrictive assumptions in many control applications as information my be distributed among agents who must communicate with another and agents may have inherent dynamics, e.g., a group of mobile robots. Any GNE seeking agorithm that is applied in these scenarios must deal with these two issues. Recently, NE seeking with partial-decision information has been considered for networks of dynamic agents (multi-integrator, LTI) %[dePersis, RomanoPavel]
\cite{guoDePersis2019}-\cite{romanoPavel2019}, but most existing results are restricted to games with decoupled constraints. Considering games with coupled constraints, results exist for integrator agents, e.g., \cite{pavel2019} in discrete-time and \cite{lu2019}-\cite{ZOU2021109535} in continuous-time and for multi-integrator agents in continuous-time \cite{bianchiGrammatico2019}. Semi-decentralized methods have also been considered \cite{dePersisGrammatico2018}  as well as partial-information on the dual variables only \cite{pengPavel2019}. %[DT:Fachinei], [DT: PengPavel, CT:Hong, GrammdePersis, Ren, Hu] 
%\cite{pavel2019}-\cite{dePersisGrammatico2018},(partial-info). 
However, existing methods, such as primal-dual methods, ensure that the coupled constraints are satisfied in steady-state only. These algorithms are not applicable in real-world applications where the constraints must be satisfied for all time, e.g., sensor networks \cite{sjs12}, demand-side management in smart-grids \cite{mwjsl10}, or optical networks \cite{lp06}.%[REF, sensor  networks, smart-grid, optical networks].

In this paper, we investigate an inexact-penalty based dynamics for GNE seeking with passive agents in networks. These dynamics can ensure that the coupled constraints are satisfied for all time, not just in steady-state. Moreover, this approach allows for the extension from the full-decision information setting to the partial-information one.

Penalty methods, using exact penalty functions, have been used in GNE seeking algorithms \cite{facchineiKanzow2010b}-\cite{fukushima2011}. In our work, in order to enforce constraint satisfaction, the penalty function takes the form of a barrier function that prevents each agent's action from exiting the interior of the constraint set. Related to our work, \cite{fabiani2019} considers NE seeking in potential games with virtual couplings used to satisfy connectivity constraints for all time. Compared to \cite{fabiani2019}, we allow for arbitrary convex inequality constraints in non-potential games. As far as we are aware, there do not exist any general GNE seeking methods that can ensure constraint satisfaction for all time.

\textit{Contributions. }
Interior point methods are a much used tool in convex optimization \cite{boyd_vandenberghe_2004}. We propose using the log-barrier function on the coupled inequality constraints as our inexact penalty function. The GNE problem is then converted into an NE problem whose costs go to infinity at the boundary of the constraint set. We consider GNE seeking using this new set of penalized cost functions for nonlinear agents with a class of equilibrium-independent passive (EIP) dynamics. The benefit of considering these agents is two-fold. First, we are able to capture versions of a variety of already known NE seeking algorithms. Secondly, we are able to consider certain types of dynamic agents. For the full-information case a gradient based feedback is used. In partial information setting, we instead use a Laplacian based feedback, \cite{gpTAC18}, for the case where agents have full knowledge of the constraint information.

In both these cases, assuming that the initial conditions satisfy the constraints, we show that the resulting action trajectories also satisfy the constraints for all time. Additionally, we extend these results to cases where the agents do not have full-information about the constraints, but must communicate in order to get it using a two-time scale approach.

A preliminary version of this work appeared in \cite{romanopavel2020} concerning systems with LTI passive dynamics. Only cases with full-information of the constraint info are treated therein.

This paper is organized as follows: In Section \ref{sec:background}, we provide the necessary background information on invariance, passivity and graph theory. In Section \ref{sec:problem}, we formulate the inexact penalized GNE problem. In Section \ref{sec:nonlinearFullInfo}, we provide a GNE seeking strategy for agents' with EIP dynamics under full-information of the others' actions. In Section \ref{sec:nonlinearpartInfo}, these results are extended to the case of partial action information but full-knowledge of the constraint. Section \ref{sec:fullyDist} considers a fully-distributed algorithm based on a two-time scale approach. Sections \ref{sec:OSNR} and \ref{sec:robots} provide applications to optical networks and velocity synchronization respectively.

\emph{Notations}: Let $\reals$ and $\reals_{\geq 0}$denote the real numbers and non-negative real numbers, respectively. Given $x,y \in \reals^n$, $x\top y$ denotes their inner product. Let $\|\cdot\|:\reals^n \rightarrow \real_{\geq 0}$ denote the Euclidean norm and $\|\cdot\|:\reals^n \rightarrow \reals_{\geq 0}$ its induced matrix norm. Given a set $\mathcal{X} \subset \reals^n$, $\|\cdot\|_{\mathcal{X}}:\reals^n \rightarrow \real_{\geq 0}$ denotes the Euclidean point-to-set distance. Given a function $\phi:\reals^n \rightarrow \reals$ and a vector field $f:\reals^n \rightarrow \reals^n$, $L_f \phi = \nabla \phi^\top f$ is the Lie-derivative of $\phi$ along $f$.

\section{Background}\label{sec:background}
\subsection{Positive- and Output-Positive-Invariance}
Consider a system with dynamics given by
\vspace{-0.2cm}
\begin{align} \label{eq:positiveInv}
\begin{split}
\dot x &= f(x)\\
y&=h(x)
\end{split}
\end{align}
where $x\in \reals^n$ and $y\in \reals^m$ and $f:\reals^n\rightarrow \reals^n$ is locally Lipschitz. The following are some standard results concerning positive-invariance of sets from, e.g., \cite{blanchini_2015}.
\begin{defn}
	A set $\mathcal{X}\subset \reals^n$ is called positively-invariant for \eqref{eq:positiveInv} if for all $x(0)\in \mathcal{X}$, $x(t)\in \mathcal{X}$ for all $t \geq 0$.
\end{defn}
\begin{defn}
	A set $\mathcal{Y}\subset \reals^m$ is called output-positively-invariant \eqref{eq:positiveInv} if $y(0)\in \mathcal{Y}$ implies that $y(t)\in \mathcal{Y}$ for all $t \geq 0$.
\end{defn}
\begin{defn}
	The Bouligand tangent cone of the set $\mathcal{X}$ at $x$ is $T_\mathcal{X}(x) = \{v\in \reals^n: \lim\inf_{t\rightarrow 0^+} \frac{\|x+tv\|_\mathcal{X}}{t}=0 \}$
\end{defn}
\begin{defn} (Definition 4.9 \cite{blanchini_2015})\label{def:practicalSet}
	Let $\mathcal{O}$ be an open set. A set $S \subset \mathcal{O}$ is a practical set if
	\begin{enumerate}
		\item $S$ is defined by a finite set of inequalities\vspace{-0.2cm}
		\begin{align*}
		S = \{x \in \reals^n: h_k(x) \leq 0, k = 1,\dots,r\}
		\end{align*}
		where $h_k(x)$ are continuously differentiable functions defined on $\mathcal{O}$.
		\item For all $x \in S$, there exists $z$ such that\vspace{-0.2cm}
		\begin{align*}
		h_k(x) + \nabla h_k(x)^\top z <0, \quad \forall \  k
		\end{align*}
		\item There exists a Lipschitz continuous vector field $f_0(x)$ such that for all $x \in \partial S$,\vspace{-0.2cm}
		\begin{align*}
		L_{f_0} h_k < 0, \quad \forall \ k \ \st \ h_k(x)=0
		\end{align*}
	\end{enumerate}
\end{defn}
\begin{lemma} \label{thm:invariant}
	A closed set $\mathcal{X}\subset \reals^n$ is positively-invariant if and only if $f(x)\in T_\mathcal{X}(x)$ for all $x\in \mathcal{X}$.
\end{lemma}
\begin{lemma} \label{thm:tangent}
	Let $S\!\subset\! \mathcal{O}$ be a practical set. Then for all $x\! \in\! \partial S$,
	\begin{align*}
	T_S(x) = \{ z \in \reals^n : L_z h_k(x) \leq 0, \forall \ k \ \st \ h_k(x)=0 \}.
	\end{align*}
\end{lemma}

\subsection{Passivity and Observability}
The following is from \cite{porco2019}. Consider a system
\vspace{-0.2cm}
\begin{align} \label{eq:EIPsys}
\Xi:\begin{cases}
\dot x = f(x) +Gu\\
y = h(x)
\end{cases}
\end{align}
where $G$ is full column-rank. Let $\mathcal{E}_{\Xi}$ denote the set of assignable equilibria of \eqref{eq:EIPsys}, and given $\bar x \in \mathcal{E}_{\Xi}$, let $\bar u = k_u(\bar x) := -(G^\top G)^{-1}G^\top f(\bar x)$ and $\bar y = h(\bar x)$ be the equilibrium input and output.
\begin{defn}
	The system \eqref{eq:EIPsys} is \emph{equilibrium-independent passive} if, for every equilibrium $\bar{x}$, there exists a continuously-differentiable storage function $V^{\bar{x}}:\reals^n \rightarrow \reals_{\geq 0}$ such that
	\vspace{-0.2cm}
	\begin{align} \label{eq:EIP}
	\nabla V^{\bar{x}}(x)^\top(f(x)+Gu) \leq (y-\bar y)^\top (u-\bar u)
	\end{align}
	where $\bar u$ and $\bar y$ are the steady-state input and output at $\bar x$. A set of storage functions $\{V^{\bar x}(x), \bar x \in \mathcal{E}_{\Xi}\}$ satisfying \eqref{eq:EIP} is an \emph{EIP storage function family}.
\end{defn}

\begin{defn}
	The system \eqref{eq:EIPsys} is \emph{equilibrium-independent observable} if, for every $\bar x \in \mathcal{E}_{\Xi}$ with associated equilibrium input/output vectors $\bar u = k_u(\bar x)$ and $\bar y = h(\bar x)$, no trajectory of $\dot x = f(x) + G\bar u$
	can remain within the set $\{x \in \reals^n: h(x) = \bar y\}$ other
	than the equilibrium trajectory $x(t) = \bar x$.
\end{defn}

\begin{lemma} \label{lemma:EqUnique}
	If the system \eqref{eq:EIPsys} is equilibrium-independent observable then for a given equilibrium I/O pair $(\bar u, \bar y)$, there is exactly one $\bar x \in \mathcal{E}_\Sigma$ satisfying $\bar u = k_u(\bar x)$ and $\bar y = h(\bar x)$.
\end{lemma}

\subsection{Graph Theory}
The following is from \cite{gr01}. An undirected graph $\mathcal{G} = (\mathcal{I},\mathcal{E})$ is a set of vertices, $\mathcal{I} = \{1,\dots,N\}$, and edges, $\mathcal{E}\subset \mathcal{I}\times \mathcal{I}$. $(i,j)\in \mathcal{E}$ means that vertex $j$ can receive information from vertex $i$. Let $\mathcal{G}$ be assumed to be undirected so that $(i,j)\in \mathcal{E}$ if and only if $(j,i)\in \mathcal{E}$ for all $i,j \in \mathcal{I}$. The adjacency matrix $\boldsymbol{A} = [a_{ij}]\in \reals^{N\times N}$ of $\mathcal{G}$ is defined by $a_{ij} = 1$ if $(j,i)\in\mathcal{E}$ and $a_{ij}=0$ otherwise. Since $\mathcal{G}$ is undirected $\boldsymbol{A}=\boldsymbol{A}^\top$. $\mathcal{G}$ is connected if given any $i,j\in \mathcal{I}$, there is a path connecting them. Let $\mathcal{N}_i$ denote the neighbours of vertex $i$ and let $D = \diag(|\mathcal{N}_i|)_{i\in\mathcal{I}}$. The Laplacian of $\mathcal{G}$ is defined as $L = D-\boldsymbol{A}$.

\section{Problem Formulation} \label{sec:problem}
In this work, we consider a set of $N$ agents (players) $\mathcal{I}=\{1,\dots,N\}$ in a generalized game. Each agent controls its action $y_i\in \reals^{m_i}$ and attempts to minimize a cost function $\mathcal{J}_i$ subject to shared, (possibly) coupled inequality constraints $g(y)\leq 0$, $g:\reals^m \rightarrow \reals^p$, where $y = \col(y_i)_{i\in \mathcal{I}}\in \reals^m$, $m=\sum_{i\in \mathcal{I}}m_i$. This gives the following for each $i\in \mathcal{I}$:
\vspace{-0.2cm}
\begin{align} \label{eq:GNE}
\begin{split}
\min_{y_i}\quad &\mathcal{J}_i(y_i,y_{-i})\\
\st\quad  &g(y)\leq 0 
\end{split}
\end{align}
where $y_{-i} = \col(y_1,\dots,y_{i-1},y_{i+1},\dots,y_N)$, all agents actions except for agent $i$'s.
\begin{defn}
	A \emph{generalized Nash equilibrium} (GNE) of \eqref{eq:GNE} is a strategy profile $y^*$ satisfying
	\vspace{-0.2cm}
	\begin{align*}
	y_i^* \in \argmin_{z_i} \mathcal{J}_i (z_i,y^*_{-i})\ \text{s.t.} \ g(z_i,y^*_{-i})\leq 0, \quad \forall i \in \mathcal{I} 
	\end{align*}
\end{defn}

\begin{asmp} \label{asmp:cost}
	The cost function of each agent, $\mathcal{J}_i(y_i,y_{-i})$, is convex and continuously-differentiable in $y_i$, for each fixed $y_{-i}$, $g(y)$ is component-wise convex and $C^2$ in $y$, either all $g_\ell(y)$ are affine or at least one $g_\ell(y)$ is strictly-convex, where $g_\ell(y)$ is the $\ell^{\text{th}}$ component of $g(y)$, and the feasible set $\Omega=\{y\in \reals^m:g(y)\leq 0\}$ is non-empty, convex, compact and satisfies Slater's constraint qualification. 
\end{asmp}
%Under Assumption \ref{asmp:cost},  in \cite{facchineiPang2007}, the GNE problem \eqref{eq:GNE} has a solution.

Let the stacked vector of partial gradients of all cost functions, called the pseudo-gradient, be denoted as
\vspace{-0.2cm}
\begin{align}
F(y) := \col(\nabla_i \mathcal{J}_i (y_i,y_{-i}))_{i \in \mathcal{I}},
\end{align}
where $\nabla_i \mathcal{J}_i (y_i,y_{-i})\!=\! \frac{\partial \mathcal{J}_i (y_i,y_{-i})}{\partial y_i}^\top\!$.

A specific type of GNE is the so-called variational-GNE (vGNE). Under Assumption \ref{asmp:cost}, from Theorem 4.8 in \cite{facchineiKanzow2010}, $y^*$ is a vGNE if and only if there exist a dual variable $\lambda^*\in \reals^p$ such that the KKT conditions hold:\vspace{-0.2cm}
\begin{align} \label{eq:vGNE}
\begin{split}
&F(y^*)+ \sum_{\ell=1}^p \lambda_\ell^*\nabla g_\ell(y^*) = 0\\
&\lambda_\ell^* \geq 0\\
&-g_\ell(y^*) \geq 0, \qquad \forall \ell =1,\dots,p\\
&-\lambda_\ell^*g_\ell(y^*) = 0
\end{split}
\end{align}
A vGNE has the interpretation of no price-discrimination among the agents, that is, all agents are penalized equally for constraint violation.
\begin{asmp} \label{asmp:pseudo}
	The pseudo-gradient $F(y)$ is Lipschitz continuous, i.e., $\|F(y)-F(y')\|\leq\theta_1\|y-y'\|$ for $\theta_1>0$, for all $y,y'\in \reals^m$, and is either
	\begin{enumerate}[a)]
		\item strictly-monotone, i.e., $(y\!-\!y')^\top\!(F(y)\!-\!F(y'))\!\!>0$ for all $y\neq y' \in \reals^m$, or
		\item strongly-monotone, i.e., $(y\!-\!y')^\top\!(F(y)\!-\!F(y'))\!\geq\!\frac{\mu}{2}\!\|y\!-\!y'\|^2$ for $\mu>0$ and for all $y,y' \in \reals^m$.
	\end{enumerate}
\end{asmp}
Assumptions \ref{asmp:cost} and \ref{asmp:pseudo} are standard assumptions that guarantee existence and uniqueness of a vGNE and are commonly used to show convergence, see, e.g., \cite{pavel2019} or \cite{bianchiGrammatico2019}. Under Assumptions \ref{asmp:cost} and \ref{asmp:pseudo}, by Cor. 2.2.5 and Thm. 2.3.3 from \cite{facchineiPang2007}, the GNE problem \eqref{eq:GNE} has a unique variational-GNE.

Our goal is to design a GNE seeking algorithm in continuous-time that satisfies the constraint $g(y)\leq0$ for all time. Inspired by interior-point methods, we consider solving problem \eqref{eq:GNE} using inexact penalty functions. Thus, consider transforming the GNE problem into an NE seeking problem given by the following set of unconstrained programs:\vspace{-0.2cm}
\begin{align} \label{eq:NE}
\min_{y_i}\quad &\mathcal{J}_i(y_i,y_{-i})+\phi(y), \quad \forall i \in \mathcal{I}
\end{align}
where $\phi(y)$ is the so-called log-barrier function
\vspace{-0.2cm}
\begin{align} \label{eq:logbarrier}
\phi(y) = \begin{cases}
-\rho \sum_{\ell=1}^p \log(-g_\ell(y)),\quad &g(y)<0\\
+\infty,&\text{else}
\end{cases} 
\end{align}
where $\rho>0$. Note that under Assumption \ref{asmp:cost}, $\phi(y)$ is strictly-convex and $C^2$ in $y$ on $\interior \Omega = \{y \in \reals^m: g(y)<0 \}$ and thus $\nabla \phi(y)$ is locally Lipschitz on $\interior \Omega$.

Under Assumptions \ref{asmp:cost}-\ref{asmp:pseudo}, by Cor. 4.3 from \cite{bo99} and Thm. 3 from \cite{sfpp14}, the NE problem \eqref{eq:NE} has a unique solution $y^*$. Moreover, $y^*$ satisfies
\vspace{-0.2cm}
\begin{align} \label{eq:NEcondition}
F(y^*)+\nabla \phi(y^*) = 0
\end{align}
where $\nabla \phi$ is the gradient of $\phi$. Following from Section 11.2.2 in \cite{boyd_vandenberghe_2004}, since $\phi(y)$ is the log-barrier function, \eqref{eq:NEcondition} becomes
\vspace{-0.2cm}
\begin{align*}
F(y^*)+ \sum_{\ell=1}^p \frac{\rho}{-g_\ell(y^*)}\nabla g_\ell(y^*) = 0
\end{align*}
By letting $\lambda_\ell^*(\rho) := \frac{\rho}{-g_\ell(y^*)}$, we get the following conditions\vspace{-0.2cm}
\begin{align*}
&F(y^*)+ \sum_{\ell=1}^p \lambda_\ell^*(\rho)\nabla g_i(y^*) = 0\\
&\lambda_\ell^*(\rho) \geq 0\\
&-g_\ell(y^*) \geq 0, \qquad \quad \forall l =1,p\\
&-\lambda_\ell^*(\rho)g_\ell(y^*) = \rho
\end{align*}
and as $\rho\!\rightarrow\! 0$, we recover \eqref{eq:vGNE}. Thus the NE of \eqref{eq:NE} is an approximate vGNE of \eqref{eq:GNE}. Moreover, it is an $\varepsilon$-GNE of \eqref{eq:GNE}.

\begin{defn}
	A \emph{generalized} $\varepsilon$-\emph{Nash equilibrium} ($\varepsilon$-GNE) of \eqref{eq:GNE} is a strategy profile $y^*$ satisfying\vspace{-0.2cm}
	\begin{align*}
	\mathcal{J}_i(y^*_i,y^*_{-i})\!\leq \!\inf_{y_i} \{\mathcal{J}_i(y_i,y_{-i}^*)\!+\!\varepsilon : g(y_i,y_{-i}^*)\!\leq \!0 \}, \, \, \forall i \in \mathcal{I}
	\end{align*}
	for $\varepsilon>0$. When $\varepsilon=0$, $y^*$ is a GNE. Moreover, if the dual variables are the same for each agent, we call $y^*$ an $\varepsilon$-vGNE.
\end{defn}

\begin{lemma} \label{lemma:epsilon}
	The NE of \eqref{eq:NE} is an $\varepsilon$-vGNE of \eqref{eq:GNE} with $\varepsilon=p\rho$.
\end{lemma}
\begin{proof}
	Consider the dual function, $h_i$, of \eqref{eq:GNE} for agent $i$ evaluated at the NE of $\eqref{eq:NE}$ $y^*$ and $\lambda^*(\rho)$ as above,\vspace{-0.2cm}
	\begin{align*}
	h_i(\lambda^*(\rho)) &= \mathcal{J}_i(y_i^*,y_{-i}^*)+\sum_{\ell=1}^{p}\lambda_\ell^*(\rho)g_\ell (y^*)\\
	&= \mathcal{J}_i(y_i^*,y_{-i}^*)-p\rho
	\end{align*}
	By properties of the dual function, Section 5.1.3 in \cite{boyd_vandenberghe_2004},\vspace{-0.2cm}
	\begin{align*}
	\mathcal{J}_i(y_i^*,y_{-i}^*)-p\rho \leq \inf_{y_i} \{\mathcal{J}_i(y_i,y_{-i}^*):g(y_i,y_{-i}^*)\leq 0 \}
	\end{align*}
	Therefore,\vspace{-0.2cm}
	\begin{align*}
	\mathcal{J}_i(y_i^*,y_{-i}^*) \leq \inf_{y_i} \{\mathcal{J}_i(y_i,y_{-i}^*)+p\rho:g(y_i,y_{-i}^*)\leq 0 \}
	\end{align*}
	Since this holds for all $i\in\mathcal{I}$, $y^*$ is an $\varepsilon$-GNE of \eqref{eq:GNE} with $\varepsilon=p\rho$. Moreover, the dual variables, $\lambda_\ell^*(\rho)$ are the same for all agents. Therefore $y^*$ is an $\varepsilon$-vGNE.
\end{proof}

\section{Full-Information Gradient Feedback} \label{sec:nonlinearFullInfo}
We consider GNE seeking for a class of passive nonlinear agents with dynamics of the form
\vspace{-0.2cm}
\begin{align} \label{eq:Plant}
\mathcal{P}_i:\begin{cases}
\dot x_i = f_i(x_i) +G_iu_i\\
y_i = G_i^\top \nabla V_i(x_i).
\end{cases}
\end{align}
where $x_i \in \reals^{n_i}$ is the state of agent $i$, $u_i\in \reals^{m_i}$ is agent $i$'s input, and the \textit{output} $y_i\in \reals^{m_i}$ of $\mathcal{P}_i$ \textit{is the action of agent} $i$. In this section, we assume that each agent $i$ has full-knowledge of all other agents' actions, $y_{-i}$, but no knowledge of their states, $x_j$, or cost functions, $\mathcal{J}_j$, for $j\neq i$. This is what we call the \emph{full-(decision) information} case.
\begin{asmp} \label{asmp:nonlinearPlant}
	$\mathcal{P}_i$ satisfies the following:
	\vspace{-0.1cm}
	\begin{enumerate}[a)]
		\item $f_i(x_i)$ is Lipschitz continuous,
		\item \eqref{eq:Plant} is equilibrium independent observable,
		\item $G_i$ is full column-rank,
		\item there exists a map $\pi_i:\reals^{m_i}\rightarrow \reals^{n_i}$ that solves the regulator equations for any $\bar y_i \in \reals^{m_i}$\vspace{-0.2cm}
		\begin{align*}
		0 &= f_i(\pi_i(\bar y_i))\\
		0&= G_i^\top\nabla V_i(x_i)-\bar y_i,
		\end{align*}
		i.e., any $\bar y_{i} \in \reals^{m_i}$ is an equilibrium output of \eqref{eq:Plant} with equilibrium input $\bar u_{i}=0$.
		%\item $\mathcal{P}_i$ is passive with positive-definite storage function $\frac{1}{2}x_i^\top P_ix_i$.
		\item $V_i(x_i)$ is a strongly convex function with $\nabla V_i(x_i)$ Lipschitz continuous and the mapping $-f_i \circ \nabla V_i^{-1}$ is monotone.
	\end{enumerate}
\end{asmp}
\begin{remark}
	Under Assumption \ref{asmp:nonlinearPlant}, by Corollary 3.6 in \cite{porco2019}, \eqref{eq:Plant} is equilibrium-independent passive with storage function family $\{V^{\bar x}(x), \bar x \in \mathcal{E}_{\mathcal{P}_i}\}$, where \vspace{-0.2cm}
	\begin{align*}
	V_i^{\bar x_i}(x_i) := V_i(x_i)-V(\bar x_i) - \nabla V_i(\bar x_i)^\top(x_i-\bar x_i).
	\end{align*}
\end{remark}

Three examples of systems that can meet Assumption \ref{asmp:nonlinearPlant} are:
\begin{enumerate}[i)]
	\item Integrators\vspace{-0.2cm}
	\begin{align} \label{eq:int}
	\dot y_i = u_i,
	\end{align}
	\item PI controllers in cascade with certain stable linear systems\vspace{-0.2cm}
	\begin{align} \label{eq:cascadePI}
	\begin{split}
	&\dot x_i = \begin{bmatrix} -v_iI&k_i I\\0&0\end{bmatrix}x_i+\begin{bmatrix}I\\I\end{bmatrix}u_i\\
	&y_i = \begin{bmatrix} I& 0\end{bmatrix}x_i,
	\end{split}
	\end{align}
	where $0<k_i<v_i$, and
%	\item Double-integrators\vspace{-0.2cm}
%	\begin{align} \label{eq:doubleInt}
%	\mathcal{P}_i:\begin{cases}
%	\dot y_i = v_i\\
%	\dot v_i = u_i,
%	\end{cases}
%	\end{align}
%	by first using the feedback transformation $u_i = -\frac{1}{b_i}v_i+\tilde u_i$ and a new choice of output $\gamma_i = y_i+b_iv_i$ where $b_i>0$, to render $\mathcal{P}_i$ passive, giving\vspace{-0.2cm}
%	\begin{align} \label{eq:doubleInt1}
%	\tilde{\mathcal{P}}_i:\begin{cases}
%	\dot y_i = v_i\\
%	\dot v_i = -\frac{1}{b_i}v_i + \tilde u_i\\
%	\gamma_i = y_i+b_i v_i.
%	\end{cases}
%	\end{align}
%	This feedback transformation is easily extended to higher-order integrators, see \cite{romanoPavel2019}.
\end{enumerate}

By stacking the plant dynamics \eqref{eq:Plant}, we get\vspace{-0.2cm}
\begin{align} \label{eq:PlantStacked}
\mathcal{P}:\begin{cases}
\dot x = f(x) +Gu\\
y = G^\top \nabla V(x).
\end{cases}
\end{align}
where $x\! = \!\col(x_i)_{i \in \mathcal{I}}$,  $u\! =\! \col(u_i)_{i \in \mathcal{I}}$,  $y\! =\! \col(y_i)_{i \in \mathcal{I}}$, $f(x) \!=\! \col(f_i(x_i))_{i \in \mathcal{I}}$, $G\! =\! \blkdiag(G_i)_{i \in \mathcal{I}}$ and $\nabla V \!=\! \col(\nabla V_i(x_i))_{i \in \mathcal{I}}$. 

\begin{remark}
It can be easily verified that \eqref{eq:PlantStacked} satisfies
	\begin{enumerate}[a)]
	\item $f(x)$ is Lipschitz continuous,
	\item \eqref{eq:PlantStacked} is equilibrium-independent observable,
	\item $G$ is full column-rank,
	\item the map $\pi:\reals^{m}\rightarrow \reals^{n}$, $\pi(x) = \col(\pi_i(x_i))_{i \in \mathcal{I}}$ solves the regulator equations for any $\bar y \in \reals^{m}$\vspace{-0.2cm}
	\begin{align*}
	0 &= f(\pi(\bar y))\\
	0&= G^\top\nabla V(x)-\bar y,
	\end{align*}
	i.e., any $\bar y \in \reals^{m}$ is a steady-state output of \eqref{eq:Plant} with steady-state input $\bar u=0$.
	%\item $\mathcal{P}_i$ is passive with positive-definite storage function $\frac{1}{2}x_i^\top P_ix_i$.
	\item \eqref{eq:PlantStacked} is equilibrium independent observable with strongly convex storage function family\vspace{-0.2cm}
	\begin{align*}
	V^{\bar x}(x) := V(x)-V(\bar x) - \nabla V(\bar x)^\top(x-\bar x)
	\end{align*}
	where $V(x) = \sum_{i \in \mathcal{I}}V_i(x_i)$.
\end{enumerate}
\end{remark}

We consider solving GNE seeking \eqref{eq:GNE} for \eqref{eq:PlantStacked} by converting the problem into \eqref{eq:NE}. In the full-(decision) information case, we consider a static partial-gradient feedback. Thus, agent $i$ with $\mathcal{P}_i$  \eqref{eq:Plant} takes\vspace{-0.2cm}
\begin{align} \label{eq:nonlinFIfeedback}
u_i = -\nabla_i \mathcal{J}_i(y)-\nabla_i \phi(y). 
\end{align}
where $\nabla_i \phi(y) = \frac{\partial \phi(y)}{\partial y_i}^\top$. For all agents, from $\mathcal{P}$ \eqref{eq:PlantStacked} with $u\!=-\!F(y)-\nabla \phi(y)$, this leads to an overall stacked dynamics, \vspace{-0.2cm}
\begin{align} \label{eq:FIdynamics}
\Sigma:\begin{cases}
\dot x = f(x)-G(F(y)+\nabla \phi(y))=:q(x)\\
y = G^\top \nabla V(x)
\end{cases}
\end{align}

\begin{remark} \label{rem:methods}
	In the \textit{absence of constraints (no penalty)}, the dynamics \eqref{eq:FIdynamics} capture a number of previously investigated NE seeking algorithms. If each $\mathcal{P}_i$ is taken to be an integrator \eqref{eq:int}, then gradient-play \vspace{-0.2cm}
	\begin{align*}
	\Sigma:\quad \dot y + F(y)=0
	\end{align*}
	is recovered \cite{gpTAC18}. If instead, each $\mathcal{P}_i$ is governed by \eqref{eq:cascadePI}, then the dynamics become the following second-order method investigated in the optimization literature \vspace{-0.2cm}
	\begin{align*}
	\Sigma:\quad \ddot{y}+\Big(\mathcal{V}+\frac{\partial F(y)}{\partial y}\Big)\dot y + KF(y)=0
	\end{align*}
	where $\mathcal{V}=\blkdiag(v_iI)_{i\in \mathcal{I}}$ and $K=\blkdiag(k_iI)_{i\in \mathcal{I}}$. See, e.g., \cite{attouch2014}-\cite{bot2019}. 
%	Finally, if we consider that each $\mathcal{P}_i$ is described by double-integrator dynamics \eqref{eq:doubleInt1} with outer-loop feedback $\tilde u_i = -\nabla_i \mathcal{J}_i(\gamma_i,\gamma_{-i})$, we get the dynamics \vspace{-0.2cm}
%	\begin{align*} 
%	\Sigma:\begin{cases}
%	\dot y = v\\
%	\dot v = -\mathcal{B}^{-1}v-F(y+\mathcal{B}v).
%	\end{cases}
%	\end{align*}
%	where $\mathcal{B} = \blkdiag(b_iI_{n_i})_{i\in\mathcal{I}}$, which were investigated in \cite{romanoPavel2019}.
\end{remark}

Next, we investigate the behaviour of \eqref{eq:FIdynamics} and show that not only do the solutions converge to an equilibrium corresponding to the output being the NE, but that the solutions satisfy the output constraints for all time. First, we investigate the equilibria of \eqref{eq:FIdynamics}.

\begin{lemma} \label{lemma:FIEq}
	Under Assumptions \ref{asmp:cost}, \ref{asmp:pseudo} and \ref{asmp:nonlinearPlant}, $x^* = \pi(y^*)$ is the unique equilibrium point of \eqref{eq:FIdynamics} with $y^*$ as in \eqref{eq:NEcondition}, the NE of \eqref{eq:NE} and $\varepsilon$-vGNE of \eqref{eq:GNE}.
\end{lemma}
\begin{proof}
See Appendix \ref{appendix:FIEq}.
\end{proof}

In order to show that the constraints are satisfied for all time, we look at positive-invariance of sub-level sets of the form $S\!=\!\{x\!\in\! \reals^n: V^{x^*}(x)\!-\!c\! \leq\! 0, \phi^{x^*}(x)\!-\!d\! \leq\! 0\}$ for $c,d>0$, where $x^*\! =\! \pi(y^*)$ and $\phi^{x^*}(x) \!:=\! \phi(G^\top \nabla V(x))\!-\!\phi(G^\top \nabla V(x^*))\!-\!\nabla \phi(G^\top \nabla V(x^*))^\top(G^\top\nabla V(x)\!-\!G^\top \nabla V(x^*))$, related to the Bregman divergence of $\phi(y)$. The interior of these sets corresponds to the output satisfying the constraints. The following lemma shows that $S$ is a practical set.
\begin{lemma} \label{lemma:practicalSet}
	Under Assumption \ref{asmp:nonlinearPlant}, $S=\{x\in \reals^n: V^{x^*}(x)-c \leq 0, \phi^{x^*}(x)-d \leq 0\}$ is a compact, practical set for all $c,d>0$.
\end{lemma}

\begin{proof}
See Appendix \ref{appendix:practicalSet}.
\end{proof}

\begin{remark}
	With LTI agents as in \cite{romanopavel2020}, sets of the form $S$ are convex with non-empty interior. Thus practicality follows trivially. In potential games, as in \cite{fabiani2019}, the sublevel sets of the potential function are positively invariant under gradient feedback and thus constraint satisfaction also follows trivially.
\end{remark}

\begin{lemma} \label{lemma:FI}
	Under Assumptions \ref{asmp:cost}, \ref{asmp:pseudo}(a) and \ref{asmp:nonlinearPlant}, the set $\interior \Omega=\{y:g(y)<0\}$ is output-positively-invariant for the dynamics \eqref{eq:FIdynamics}. That is, for all $x(0)$ such that $y(0)\in\interior\Omega$, $y(t)\in \interior\Omega$ for all $t\geq0$, i.e., the output constraints are satisfied for all time.
\end{lemma}
\begin{proof}
	We show that for each $x(0)$ such that $y(0)= G^\top \nabla V(x(0))$, there exists $c,d>0$ such that $x(0) \in S = \{x \in \reals: V^{x^*}(x) - c \leq 0\} \cap \{x \in \reals: \phi^{x^*}(x) - d \leq 0\}$, where $x^*$ as in Lemma \ref{lemma:FIEq}, and $S$ is positively invariant.
	
	First, consider the Lie derivatives of $V^{x^*}(x)$ and $\phi^{x^*}(x)$ on $\partial S$ along the solutions of \eqref{eq:FIdynamics}. There are two cases:
	\begin{enumerate}
		\item $V^{x^*}(x)=c$ and $\phi^{x^*}(x)\leq d$
		
		In this case, we consider the Lie derivative of $V^{x^*}$ along the solutions of \eqref{eq:FIdynamics}. By equilibrium-independent passivity of \eqref{eq:PlantStacked}, we have \vspace{-0.2cm}
		\begin{align*}
		L_q V^{x^*} \leq -(y\!-\!y^*)^\top u = -(y\!-\!y^*)^\top(F(y)\!+\!\nabla \phi(y)).
		\end{align*}
		Using \eqref{eq:NEcondition}, $F(y)$ monotone and $\phi(y)$ convex, we have \vspace{-0.2cm}
		\begin{align*}
		L_q V^{x^*}&\leq -(y\!-\!y^*)^\top(F(y)\!-\!F(y^*)+\nabla \phi(y)\!-\!\nabla \phi(y^*))\\
		&\leq 0.
		\end{align*}
		\item $V^{x^*}\leq c$ and $\phi^{x^*}(x)=d$
		
		We take the Lie derivative of $\phi^{x^*}(x)$ along the solutions of \eqref{eq:FIdynamics}, giving \vspace{-0.2cm}
		\begin{align*}
		L_q \phi^{x^*} = &(\nabla \phi(y)\!-\!\nabla \phi(y^*))^\top G^\top \nabla^2 V(x)^\top[f(x)\!-\!G(F(y)\\
		&+\nabla \phi(y))]\\
		=&(\nabla \phi(y)\!-\!\nabla \phi(y^*))^\top G^\top \nabla^2 V(x)^\top[f(x)-f(x^*)\\
		&-G(F(y)-F(y^*)+\nabla \phi(y)-\nabla \phi(y^*))].
		\end{align*}
		By Assumption \ref{asmp:nonlinearPlant}(c) and (e), $G^\top \nabla^2 V(x) G \succeq \alpha I$ for some $\alpha>0$. Then, we can get \vspace{-0.2cm}
		\begin{align*}
		L_q \phi^{x^*} \leq&\|\nabla \phi(y)\!-\!\nabla\phi(y^*)\|\Big[ \|G^\top \nabla^2 V(x)\|\\
		&\|f(x)\!-\!f(x^*)\|+\|G^\top \nabla^2 V(x) G\|\\
		&\|F(y)-F(y^*)\|-\alpha\|\nabla \phi(y)-\nabla \phi(y^*)\|\Big].
		\end{align*}
		Since $V^{x^*}<c$ and $V^{x^*}$ strongly convex, we have $\|x-x^*\|\leq c_1$, for some $c_1$. Since $f(x)$ and $\nabla V(x)$ are Lipschitz continuous, we have $\|f(x)-f(x^*)\|\leq \theta_3\|x-x^*\|$ and $\|G^\top \nabla V^2(x)\|\leq \theta_4$, for some $\theta_3,\theta_4$. Since $F(y)$ is Lipschitz continuous, we have $\|F(y)-F(y^*)\|\leq \theta_1\|y-y^*\|\leq\theta_1\theta_4\|x-x^*\|\leq c_1\theta_1\theta_4$. Therefore, \vspace{-0.2cm}
		\begin{align*}
		L_q \phi^{x^*} &\leq \|\nabla \phi(y)-\nabla \phi(y^*)\|\Big[c_1\theta_3\theta_4+c_1\theta_1\theta_4^2\|G\|\\
		&\quad-\alpha\|\nabla \phi(y)-\nabla \phi(y^*)\| \Big]
		\end{align*}
		By strict-convexity of $\phi^{x^*}(x)$, we have that $\|\nabla \phi(y)-\nabla \phi(y^*)\|\geq \frac{\phi^{x^*}(x)}{\|y-y^*\|}\geq \frac{d}{c_1\theta_4}$. If we take $d \geq d_c  := \frac{c_1\theta_4}{\alpha}[c_1\theta_3\theta_4+c_1\theta_1\theta_4^2\|G\|]$, then $L_q V_2 \leq 0$.	
	\end{enumerate}
	Since $S$ is a practical set by Lemma \ref{lemma:practicalSet}, by Lemma \ref{thm:tangent}, $q(x) \in T_S(x)$ for all $x \in \partial S$. In the interior of $S$, $q(x) \in T_S(x) = \reals^n$. Therefore, by Lemma \ref{thm:invariant}, $S$ is positively invariant. For all initial conditions $x(0)$ such that $y(0)\in \interior \Omega$, take $c\geq V^{x^*}(x(0))$ and $d\geq d_c$. Then $x(t)\in S$ for all $t\geq0$ and $y(t)\in \interior \Omega$ for all $t\geq 0$ since for all $x \in S$, $y = G^\top \nabla V(x) \in \interior \Omega$.
\end{proof}

\begin{thm} \label{thm:FI}
	Under Assumptions \ref{asmp:cost}, \ref{asmp:pseudo}(a) and \ref{asmp:nonlinearPlant}, the equilibrium $x^*\!=\!\pi( y^*)$ of \eqref{eq:FIdynamics}, where  $y^*$ is the NE of \eqref{eq:NE} and $\varepsilon$-vGNE of \eqref{eq:GNE}, is asymptotically stable. Moreover, if $y(0)\in \interior \Omega$, then the constraint $y(t)\in\Omega$ is satisfied for all $t\geq0$. 
\end{thm}
\begin{proof}
	Take $V^{x^*}$ as the storage function of $\mathcal{P}$, where $x^*$ as in Lemma \ref{lemma:FIEq}. From equilibrium-independent passivity of \eqref{eq:PlantStacked}, the derivative along the solutions of \eqref{eq:FIdynamics} is  \vspace{-0.2cm}
	\begin{align*}
	\dot V^{x^*} &\leq -(y-y^*)^\top u=-(y-y^*)^\top(F(y)+\nabla\phi(y))\\
	&\leq -(y\!-\!y^*)^\top(F(y)\!-\!F(y^*)+\nabla\phi(y)\!-\!\nabla\phi(y^*))\leq 0
	\end{align*}
	and, by strict-monotonicity of $F(y)$ and convexity of $\phi(y)$, $\dot V = 0$ if and only if $y=y^*$. By equilibrium independent observability of \eqref{eq:Plant}, $x=x^*$ is asymptotically stable. By Lemma \ref{lemma:FI}, if $y(0)\in \interior \Omega$, $y(t)\in \interior \Omega$ for all $t\geq0$.
\end{proof}

\section{Partial-Information Gradient Feedback} \label{sec:nonlinearpartInfo}
Now, let's assume that each agent has only partial-information of the actions taken by the other players exchanged over an undirected, connected graph, $\mathcal{G}_c$. For now, we assume that each agent has enough knowledge of $g(y)$ in order to be able to compute its partial-gradient of the penalty function exactly. This assumption is motivated by scenarios in which the agents could have this knowledge:
\begin{enumerate}[1)]
	\item In the case of standard NE seeking, each agents constraints depend only on its own action, i.e., $g(y)=\col(g_1(y_1),\dots,g_M(y_M))$, 
	\item If each agent's constraint set depends only on the actions of its neighbours in the communication graph, $\mathcal{G}_c$, or
	\item If each agent can measure the constraint independently from the other agents. %This may be the case, e.g., in a communication network if each agent can measure total power usage on the network but cannot determine who is using what.
\end{enumerate}
In Section \ref{sec:fullyDist}, we relax this assumption.

For the individual actions, assume that each agent $i$ maintains an estimate, $\textbf{y}_{j}^i$, of the action of each agent $j$ and uses these to evaluate the partial-gradient of its original cost function instead of the true actions. Let $\textbf{y}_{-i}^i:= \col(\textbf{y}_{1}^i,\dots,\textbf{y}_{i-1}^i,\textbf{y}_{i+1}^i,\dots,\textbf{y}_{i-1}^i)$ and $\textbf{y}^i:= \col(\textbf{y}_{1}^i,\dots,\textbf{y}_{i-1}^i,y_i,\textbf{y}_{i+1}^i,\dots,\textbf{y}_{i-1}^i)$. Stacking the actions and estimates, we get $\textbf{y}_{-i}:= \col(\textbf{y}_{-1}^1,\dots,\textbf{y}_{-N}^N)$, the estimates only, and $\textbf{y}:= \col(\textbf{y}^1,\dots,\textbf{y}^N)$, the stacked actions and estimates. These actions and estimates are then be exchanged over a communication graph, $\mathcal{G}_c$ with Laplacian $L$, using a proportional consensus algorithm. 
\begin{asmp} \label{asmp:graph}
	The graph $\mathcal{G}_c$ is undirected and connected.
\end{asmp}

Let matrices $\mathcal R_i $, $\mathcal S_i $ for action and estimates selection be, \vspace{-0.2cm}
\begin{align} \label{eq:RS}
\begin{split}
\mathcal R_i &:= \begin{bmatrix} 0_{m_i\times m_{<i}} & I_{m_i} & 0_{m_i\times m_{>i}} \end{bmatrix}\\
\mathcal{S}_i &:= \begin{bmatrix} I_{m_{<i}} & 0_{m_{<i} \times m_i} & 0_{m_{<i} \times m_{>i}}\\
0_{m_{>i} \times m_{<i}} & 0_{m_{>i} \times m_i} &  I_{m_{>i}} \end{bmatrix}
\end{split}
\end{align}
with $m_{<i}:=\sum_{j<i\ j,i\in \mathcal{I}}m_j$ and $m_{>i}:=\sum_{j>i\ j,i\in \mathcal{I}}m_j$. Note that $y_i = \mathcal{R}_i\textbf{y}^i$ and $\textbf{y}_{-i}^i = \mathcal{R}_i\textbf{y}^i$.

Inspired by \cite{gpTAC18}, instead of \eqref{eq:nonlinFIfeedback}, we consider that each agent \eqref{eq:Plant} uses the following dynamic feedback \vspace{-0.2cm}
\begin{align} \label{eq:nonlinearPIFeedback}
\begin{split}
\dot{\textbf{y}}_{-i}^i &= -\mathcal{S}_i\sum_{j\in \mathcal{N}_i}(\textbf{y}^i-\textbf{y}^j)\\
u_i &= -\nabla_i \mathcal{J}(y_i,\textbf{y}^i_{-i})-\nabla_i \phi(y)-\mathcal{R}_i\sum_{j\in \mathcal{N}_i}(\textbf{y}^i-\textbf{y}^j)
\end{split}
\end{align}
where $\textbf{L} = L\otimes I$ and $\nabla_i \phi(y)$ can be computed using the information available to each agent. Note that \eqref{eq:nonlinearPIFeedback} has a gradient-play term (evaluated at estimates) and penalty term, as well as a dynamic Laplacian-based estimate-consensus component $\textbf{y}_{-i}^i$ which, in steady state, should bring all $\textbf{y}^i$ to consensus. We call the stacked vector of partial gradients evaluated at estimates $\textbf{F}(\textbf{y}):=\col(\nabla_i \mathcal{J}(y_i,\textbf{y}^i_{-i}))_{i\in\mathcal{I}}$, the extended-pseudo-gradient. Note that $\textbf{F}(\textbf{1}\otimes y) = F(y)$ for all $y\in \reals^m$. By Lemma 3 from \cite{bianchiGrammatico2019}, under Assumption \ref{asmp:pseudo}(b), the extended pseudo-gradient is Lipschitz continuous, $\|\textbf{F}(\textbf{y})-\textbf{F}(\textbf{y}')\|\leq \theta_2 \|\textbf{y}-\textbf{y}'\|$ for all $\textbf{y},\textbf{y}'\in\reals^{M}$, $M=Nm$, for $\mu\leq \theta_2 \leq \theta_1$.

From \eqref{eq:Plant} and \eqref{eq:nonlinearPIFeedback}, this gives overall stacked dynamics of \vspace{-0.2cm}
\begin{align} \label{eq:PIDynamics}
\Sigma: \begin{cases}
\dot{\textbf{y}}_{-i} = -\mathcal{S}\textbf{L}\textbf{y}=:\boldsymbol{f}_1(\textbf{x})\\
\dot x = f(x)\! -\!G(\textbf{F}(\textbf{y})\!+\!\nabla \phi(y)\!+\!\mathcal{R}\textbf{L}\textbf{y})\!=:\!\boldsymbol{f}_2(\textbf{x})\\
y=G^\top \nabla V(x)
\end{cases}
\end{align}
where $\textbf{x}=\col(\textbf{y}_{-i},x)$, $\mathcal{R}=\blkdiag(\mathcal{R}_i)_{i\in\mathcal{I}}$ and $\mathcal{S}=\blkdiag(\mathcal{S}_i)_{i\in\mathcal{I}}$. Let $\boldsymbol{f}(\textbf{x}):=\col(\boldsymbol{f}_1(\textbf{x}),\boldsymbol{f}_2(\textbf{x}))$. The unique equilibrium point of \eqref{eq:PIDynamics} is $(\bar{\textbf{y}}_{-i},x^*)=(\mathcal{S}\textbf{1}_{N}\otimes y^*,\pi(y^*))$, where $y^*$  is the NE of \eqref{eq:NE} and $\epsilon$-vGNE of \eqref{eq:GNE}. We denote $\bar{\textbf{y}} = \textbf{1}_{N}\otimes y^*$.

\begin{remark}
	It is important to note that in the extended space, $\textbf{F}$ is not monotone which was critical in the proof of Theorem \ref{thm:FI}. As such, we need to use different analyses to examine the behaviour of \eqref{eq:PIDynamics}.
\end{remark}

\begin{lemma} \label{lemma:PIEq}
	Under Assumptions \ref{asmp:cost}, \ref{asmp:pseudo}(b), \ref{asmp:nonlinearPlant} and \ref{asmp:graph}, $(\textbf{y}_{-i},x) = (\mathcal{S}\textbf{1}_N \otimes y^*, \pi(x^*))$ is the unique equilibrium of \eqref{eq:PIDynamics}.
\end{lemma}

\begin{proof}
See Appendix \ref{appendix:PIEq}.
\end{proof}

We show output-positive invariance but showing positive invariance of the sublevel sets of two functions:
\begin{align*}
\textbf{V}^{\bar{\textbf{x}}}(\textbf{x}) &:= \frac{1}{2}\|\textbf{y}_{-i}-\bar{\textbf{y}}_{-i}\|^2+V^{x^*}(x)\\
\boldsymbol{\phi}^{\bar{\textbf{x}}}(\textbf{x}) &:= \phi^{x^*}(x)
\end{align*}
First, we show that the intersection of their sub-level sets forms a practical set.
\begin{lemma} \label{lemma:practicalSet2}
	Under Assumption \ref{asmp:nonlinearPlant}, $\textbf{S}=\{\textbf{x}\in \reals^{M-m+n}: \textbf{V}^{\bar{\textbf{x}}}(\textbf{x})-c \leq 0, \boldsymbol{\phi}^{\bar{\textbf{x}}}(\textbf{x})-d \leq 0\}$ is a compact, practical set for all $c,d>0$.
\end{lemma}

\begin{proof}
	By letting $\textbf{G} = \col(0,G)$, we have that $y = \textbf{G}^\top \nabla \textbf{V}(\textbf{x})$. Additionally, since $\textbf{V}(\textbf{x})$ is strongly-convex in $\textbf{x}$, the proof follows almost identically to the proof of Lemma \ref{lemma:practicalSet}, replacing $x$ with $\textbf{x}$, $G$ with $\textbf{G}$, $V^{x^*}(x)$ with $\textbf{V}^{\bar{\textbf{x}}}(\textbf{x})$ and $\phi^{x^*}(x)$ with $\boldsymbol{\phi}^{\bar{\textbf{x}}}(\textbf{x})$ and is thus omitted for the sake of brevity.
\end{proof}	

\begin{lemma} \label{lemma:nonlinearPI}
	Under Assumptions \ref{asmp:cost}, \ref{asmp:pseudo}(b), \ref{asmp:nonlinearPlant} and \ref{asmp:graph}, if $\mu(\lambda_2(L)-\theta_2)>\theta_2^2$, then the set $\interior \Omega=\{y:g(y)<0\}$ is output-positively-invariant for the dynamics \eqref{eq:PIDynamics}. That is, for all $x(0)$ such that $y(0)\in\interior\Omega$, $y(t)\in \interior\Omega$ for all $t\geq0$, i.e., the output constraints are satisfied for all time.
\end{lemma}
\begin{proof}
 We show that for all $c>0$, there exists $d_c$ such that for all $d\geq d_c$, $\textbf{S}$ is positively invariant for \eqref{eq:PIDynamics}. Then, for all $(\textbf{y}_{-i}(0),x(0))$ such that $y(0)\in \interior \Omega$, we have that $(\textbf{y}_{-i}(t),x(t))\in S_c^d$ for some $c,d>0$ and $y(t)\in \interior \Omega$ for all $t\geq0$. The proof is similar to that of Lemma \ref{lemma:FI}, however the fact that $\textbf{F}$ is not monotone complicates the analysis.
	
	First, consider the Lie derivatives of $\textbf{V}^{\bar{\textbf{x}}}(\textbf{x})$ and $\boldsymbol{\phi}^{\bar{\textbf{x}}}(\textbf{x})$ along the solutions of \eqref{eq:FIdynamics}. There are two cases to check:
	\begin{enumerate}
		\item $\textbf{V}^{\bar{\textbf{x}}}(\textbf{x})=c$ and $\boldsymbol{\phi}^{\bar{\textbf{x}}}(\textbf{x})\leq d$
		
		In this case, we consider the Lie derivative of $V$ along the solutions of \eqref{eq:PIDynamics}. By equilibrium-independent passivity of \eqref{eq:PlantStacked}, we get \vspace{-0.2cm}
		\begin{align*}
		L_{\boldsymbol{f}} \textbf{V}^{\bar{\textbf{x}}} & \leq -(\textbf{y}_{-i}-\bar{\textbf{y}}_{-i})^\top\mathcal{S}\textbf{L}\textbf{y}-(y-y^*)^\top u\\
		&\leq -(\textbf{y}_{-i}-\bar{\textbf{y}}_{-i})^\top\mathcal{S}\textbf{L}\textbf{y}\\
		&\quad-(y-y^*)^\top(\textbf{F}(\textbf{y})+\nabla \phi(y)+\mathcal{R}\textbf{L}\textbf{y}).
		\end{align*}
		Since $\textbf{F}(\bar{\textbf{y}})=F(y^*)$, $F(y^*)+\nabla \phi(y^*)=0$ and $\textbf{L}\bar{\textbf{y}}=0$, we have \vspace{-0.2cm}
		\begin{align*}
		L_{\boldsymbol{f}} \textbf{V}^{\bar{\textbf{x}}} &\leq -(\textbf{y}_{-i}-\bar{\textbf{y}}_{-i})^\top\mathcal{S}\textbf{L}(\textbf{y}-\bar{\textbf{y}})\\
		&\quad-(y-y^*)^\top(\textbf{F}(\textbf{y})-\textbf{F}(\bar{\textbf{y}})+\nabla \phi(y)-\nabla \phi(y^*))\\
		&\quad-(y-y^*)^\top\mathcal{R}\textbf{L}(\textbf{y}-\bar{\textbf{y}}).
		\end{align*}
		By monotonicity of $\nabla \phi(y)$ and that $\mathcal{R}^\top y+\mathcal{S}^\top \textbf{y}_{-i}\! =\! \textbf{y}$, \vspace{-0.2cm}
		\begin{align}\label{eq_PI_Vdot_NL}
		L_{\boldsymbol{f}} \textbf{V}^{\bar{\textbf{x}}} &\!\leq -\!(\textbf{y}\!-\!\bar{\textbf{y}})^\top\textbf{L}(\textbf{y}\!-\!\bar{\textbf{y}})
		\!-\!(\textbf{y}\!-\!\bar{\textbf{y}})^\top\mathcal{R}^\top(\textbf{F}(\textbf{y})\!-\!\textbf{F}(\bar{\textbf{y}})).
		\end{align}
		If $\mu(\lambda_2(L)-\theta_2)>\theta_2^2$, then by Theorem 2 in \cite{gpTAC18}, we have that	$L_{\boldsymbol{f}} \textbf{V}^{\bar{\textbf{x}}}\leq 0$.
		\item $\textbf{V}^{\bar{\textbf{x}}}(\textbf{x})\leq c$ and $\boldsymbol{\phi}^{\bar{\textbf{x}}}(\textbf{x})=d$
		
		In the second case, we take the Lie derivative of $\boldsymbol{\phi}^{\bar{\textbf{x}}}(\textbf{x})$ along the solutions of \eqref{eq:FIdynamics}, giving \vspace{-0.2cm}
		\begin{align*}
		L_{\boldsymbol{f}} \boldsymbol{\phi}^{\bar{\textbf{x}}} &=(\nabla \phi(y)-\nabla \phi(y^*))^\top G^\top \nabla^2 V(x)[ f(x)\\
		&\quad-G(\textbf{F}(\textbf{y})+\nabla \phi(y)+\mathcal{R}\textbf{L}\textbf{y})].
		\end{align*}
		By Assumption \ref{asmp:nonlinearPlant} (c) and (e), we have that $G^\top \nabla^2 V(x) G \succeq \alpha I$ for some $\alpha\!>\!0$. Then, we get \vspace{-0.2cm}
		\begin{align*}
		L_{\boldsymbol{f}} \boldsymbol{\phi}^{\bar{\textbf{x}}} &=(\nabla \phi(y)-\nabla \phi(y^*))^\top G^\top \nabla^2 V(x)[ f(x)-f(x^*)\\
		&\quad-G(\textbf{F}(\textbf{y})-\textbf{F}(\bar{\textbf{y}})+\nabla \phi(y)-\nabla \phi(y^*)\\
		&\quad+\mathcal{R}\textbf{L}(\textbf{y}-\bar{\textbf{y}}))]\\
		&\leq\|\nabla \phi(y)\!-\!\nabla \phi(y^*)\|\Big[ \|G^\top \nabla^2 V(x)\|\|f(x)\!-\!f(x^*)\|\\
		&\quad +\|G^\top \nabla^2 V(x)\|\|G\|(\|\textbf{F}(\textbf{y})-\textbf{F}(\bar{\textbf{y}})\|\\
		&\quad+\|\mathcal{R}\textbf{L}\|\|\textbf{y}-\bar{\textbf{y}}\|)-\alpha\|\nabla \phi(y)-\nabla \phi(y^*)\|\Big].
		\end{align*}
		Using similar arguments to those in the proof of Lemma \ref{lemma:FI}, there exists $d_c$ such that for all $d\geq d_c$ $L_{\boldsymbol{f}} \boldsymbol{\phi}^{\bar{\textbf{x}}} \leq 0$.
	\end{enumerate}
	Then, by Lemmas \ref{lemma:practicalSet2} and \ref{thm:tangent} for all $\textbf{x}\in \partial \textbf{S}$, $\boldsymbol{f}(\textbf{x})\in T_{\textbf{S}}(\boldsymbol{x})$ and for $\textbf{x}\in \interior \textbf{S}$, $\boldsymbol{f}(\textbf{x})\in \reals^{M-m+n}= T_{\textbf{S}}(\boldsymbol{x})$. Therefore, $\textbf{S}$ is positively invariant by Lemma \ref{thm:invariant}. For all initial conditions $x(0)$ such that $y(0)\in \interior \Omega$, take $c\geq V(\textbf{y}_{-i}(0),x(0))$ and $d\geq d_c$. Then $\textbf{x}(t)\in \textbf{S}$ for all $t\geq0$ and $y(t)\in \interior \Omega$ for all $t\geq 0$.
\end{proof}
\begin{thm}
	Under Assumptions \ref{asmp:cost}, \ref{asmp:pseudo}(b), \ref{asmp:nonlinearPlant} and \ref{asmp:graph}, if $\mu(\lambda_2(L)-\theta_2)>\theta_2^2$, then the equilibrium $(\bar{\textbf{y}}_{-i},x^*)=(\mathcal{S}\textbf{1}_Ny^*,\pi(y^*))$ of \eqref{eq:PIDynamics}, where $y^*$ is the NE of \eqref{eq:NE}  and $\varepsilon$-vGNE of \eqref{eq:GNE}, is asymptotically stable. Moreover, if $y(0)\in \interior \Omega$, then the constraint $y(t)\in\Omega$ is satisfied for all $t\geq0$. 
\end{thm}
\begin{proof}
	Consider the Lyapunov candidate function $V = \frac{1}{2}\|\textbf{y}_{-i}-\bar{\textbf{y}}_{-i}\|^2+V^{x^*}(x)$, where $x^*=\pi(y^*)$ as in Lemma \ref{lemma:PIEq}. Taking the derivative along the solutions of \eqref{eq:PIDynamics} yields as in \eqref{eq_PI_Vdot_NL}, \vspace{-0.2cm}
	\begin{align*}
	\dot V \leq -(\textbf{y}-\bar{\textbf{y}})^\top\textbf{L}(\textbf{y}-\bar{\textbf{y}})-(\textbf{y}-\bar{\textbf{y}})^\top\mathcal{R}^\top(\textbf{F}(\textbf{y})-\textbf{F}(\bar{\textbf{y}}))\
	\end{align*}
	By Theorem 2 in \cite{gpTAC18}, if $\mu(\lambda_2(L)-\theta_2)>\theta_2^2$, $\dot V \leq 0$ and $\dot V = 0$ if and only if $\textbf{y}=\bar{\textbf{y}}$. By equilibrium independent observability, $(\textbf{y}_{-i},x)=(\bar{\textbf{y}}_{-i},x^*)$ is asymptotically stable. By Lemma \ref{lemma:nonlinearPI}, if $y(0)\in \interior \Omega$, $y(t)\in \interior \Omega$ for all $t\geq0$.
\end{proof}

\section{Fully-Distributed Constraint Information} \label{sec:fullyDist}
Next, we consider a partial-information feedback with fully-distributed constraint information. Accordingly, each agent maintains estimates of all other agents' actions which are exchanged over a communication graph in the same manner as \eqref{eq:nonlinearPIFeedback} and uses these and an auxiliary variable in place of its own action to compute the gradient of the penalty function. In order to maintain constraint satisfaction, the communication and action updates occur on two different time-scales.

We consider the following feedback law for dynamics \eqref{eq:Plant},\vspace{-0.2cm}
\begin{align} \label{eq:CenterFeedback}
\begin{split}
\epsilon \dot{z}_i &=y_i-z_i\\ 
\epsilon\dot{\textbf{y}}_{-i}^i &= -\mathcal{S}_i\sum_{j\in \mathcal{N}_i}(\textbf{y}^i-\textbf{y}^j)\\
u_i &= -k\Big [\nabla_i \mathcal{J}(y_i,\textbf{y}^i_{-i})+\nabla_i \phi(z_i,\textbf{y}^i_{-i}) ]
\end{split}
\end{align}
where $k>0$ is a parameter to be chosen. Stacking these together yields\vspace{-0.2cm}
\begin{align}
\label{eq:CentreDynamics2}
\Sigma: &\begin{cases}
\epsilon \dot z = y - z\\
\epsilon \dot{\textbf{y}}_{-i} = -\mathcal{S}\textbf{L}\mathcal{S}^\top\textbf{y}_{-i}-\mathcal{S}\textbf{L}\mathcal{R}^\top y\\
\dot x = f(x)\! -\!kG(\textbf{F}(y,\textbf{y}_{-i})\!+\!\boldsymbol{\Psi}(z,\textbf{y}_{-i}))\\
y=G^\top\nabla V(x)
\end{cases}
\end{align}
where $z=\col(z_i)_{i \in \mathcal{I}}$ and $\boldsymbol{\Psi}(z,\textbf{y}_{-i}) = \col(\nabla_i \phi(z_i,\textbf{y}^i_{-i}))_{i\in \mathcal{I}}$.

\begin{asmp} \label{asmp:OSSPassive}
	For each $\mathcal{P}_i$, for every $\bar x_i$, $V_i^{\bar x}$ satisfies 
	\begin{align*}
	\underline{\omega}_i \|x_i-\bar x_i\|^2 \leq V_i^{\bar x_i} &\leq \bar\omega_i \|x_i-\bar x_i\|^2\\
	(\nabla V_i^{\bar x_i})^T(f_i(x_i)+G_iu_i)&\leq -\alpha_i\|x_i-\bar x_i\|^2 + \beta_i \|y_i-\bar y_i\|^2\\
	&\quad+(y_i-\bar y_i)^\top(u_i-\bar u_i)
	\end{align*}
	for $\alpha_i,\beta_i,\underline{\omega}_i,\bar \omega_i>0$.
\end{asmp}

\begin{asmp} \label{asmp:C2}
	For each $\mathcal{P}_i$, $f_i(x_i)$ is $C^2$ in $x_i$ and $V(x)$ is $C^3$. Additionally, $\textbf{F}(y)$ is $C^2$ in its arguments.
\end{asmp}

\begin{remark}
	Assumption \ref{asmp:OSSPassive} strengthens the EIP of Assumption \ref{asmp:nonlinearPlant} and is related to strict-passivity and output-to-state stability. It guarantees that $\bar x_i$ can rendered exponentially stable by a suitable static output feedback.
\end{remark}

\begin{remark}
	Systems \eqref{eq:int} and \eqref{eq:cascadePI}, which satisfy Assumption \ref{asmp:nonlinearPlant}, also satisfy Assumption \ref{asmp:OSSPassive} for any $\beta_i>0$.
\end{remark}

\begin{lemma} \label{lemma:centreInv}
	Under Assumptions \ref{asmp:cost} and \ref{asmp:pseudo}(b)  through \ref{asmp:C2}, for every $x(0) \in \reals^n$  such that $y(0) = G^\top\nabla V(x(0)) \in \interior \Omega$ and for every $(z(0),\textbf{y}_{-i}(0)) \in R_y^{x(0)} := \{(z,\textbf{y}_{-i}):\|z-y(0)\|^2+\|\textbf{y}_{-i}-\mathcal{S}\textbf{1}_N \otimes y(0)\|^2 \leq \gamma\}$, where $\gamma > 0$ is such $R_y^{x(0)}$ is a compact subset of $\Omega_{\textbf{y}} := \{(z,\textbf{y}_{-i}):g_{\ell}(z_i,\textbf{y}_{-i}^i)<0, \forall \ell,i \}$, there exists $\epsilon^{*}>0$ and $k^*>0$ such that for all $0<\epsilon < \epsilon^{*}$ and $k\geq k^*$, $y(t) \in \interior \Omega$ for all $t\geq0$.
\end{lemma}

\begin{proof}
First, we freeze $\epsilon = 0$ and analyze the behaviours of the reduced system and boundary layer system separately. Consider the boundary layer system. On the fast time-scale, by treating $y$ as fixed and the change of coordinates $z \mapsto \tilde z := z - y$, $\textbf{y}_{-i} \mapsto \tilde{\textbf{y}}_{-i} := \textbf{y}_{-i} - \mathcal{S}\textbf{1}_N\otimes y$, it becomes\vspace{-0.2cm}
\begin{align} \label{eq:boundary}
\begin{split}
\frac{d\tilde z}{d\tau}&=-\tilde z\\
\frac{d\tilde{\textbf{y}}_{-i}}{d\tau} &= -\mathcal{S}\textbf{L}\mathcal{S}^\top \tilde{\textbf{y}}_{-i}
\end{split}
\end{align}

The is an asymptotically stable linear system and thus $(\tilde z,\tilde{\textbf{y}}_{-i}) = (0,0)$ is an exponentially stable equilibrium point of \eqref{eq:boundary}. Next, consider the reduced system \eqref{eq:CentreDynamics2}. By treating $(z,\textbf{y}_{-i}) = (y,\mathcal{S}\textbf{1}_N\otimes y)$, the $x$ dynamics \eqref{eq:CentreDynamics2} become\vspace{-0.2cm}
\begin{align} 
\label{eq:reduced}
\begin{split}
\dot x &= f(x) -kG(F(y)+\nabla \phi(y))\\
y&=G^\top\nabla V(x)
\end{split}
\end{align}
similar to \eqref{eq:FIdynamics}.

Consider the Lyapunov candidate function $V^{x^*} := V(x)-V(x^*) - \nabla V(x^*)^\top(x-x^*)$, where $x^* = \pi(y^*)$. Taking the derivative along the solutions of \eqref{eq:reduced}, by Assumption \ref{asmp:OSSPassive},\vspace{-0.2cm}
\begin{align}
\begin{split}
\dot V &\leq -\alpha \|x - x^*\|^2+\beta \|y-y^*\|^2 +(y-y^*)^\top u\\
&\leq-\alpha \|x - x^*\|^2+\beta \|y-y^*\|^2\\
&\quad -(y-y^*)^\top(F(y)-F(y^*)+\nabla \phi(y)-\nabla \phi(y^*))\\
&\leq -\alpha \|x - x^*\|^2+\beta \|y-y^*\|^2 -\!k\mu \|y-y^*\|^2
\end{split}
\end{align}
where $\alpha = \min\{\alpha_i:i\in\mathcal{I}\}$ and $\beta = \max\{\beta_i:i\in\mathcal{I}\}$. The last inequality follows from strong-monotonicity of $F(y)+\nabla \phi(y)$. If $k\geq k^*:= \frac{\beta}{\mu}$, then $\dot V \leq -\alpha \|x - x^*\|^2$.
Thus, by Theorem 4.10 in \cite{hk02}, $x=x^*$ is an exponentially stable equilibrium of \eqref{eq:reduced}.

Second, for $\Sigma$ \eqref{eq:CentreDynamics2}, we must check properties of $f_1(z,\textbf{y}_{-i},x) := \col(-z+G^\top\nabla V(x),-\mathcal{S}\textbf{L}\mathcal{S}\textbf{y}_{-i}-\mathcal{S}\textbf{L}\mathcal{R}^\top G^\top \nabla V(x))$ and $f_2(z,\textbf{y}_{-i},x) := f(x)-kG(\textbf{F}(y,\textbf{y}_{-i})+\boldsymbol{\Psi}(z,\textbf{y}_{-i}))$.  $f_1$ is $C^1$ and Lipschitz continuous in $(z,\textbf{y}_{-i},x)$ since $\nabla V(x)$ is Lipschitz continuous and $C^1$. Therefore on any compact subset of $\reals^{M} \times \Omega_x$, $\Omega_x = \{x\in \reals^n: G^\top\nabla V(x) \in \interior \Omega\}$, $f_1$ and its partial derivatives are continuous and bounded. Under Assumptions \ref{asmp:nonlinearPlant} and \ref{asmp:OSSPassive}, $f_2$ is $C^1$, locally Lipschitz with locally Lipschitz partial derivatives. Therefore, $f_2$ and its partial derivatives are continuous and bounded on any compact subset of $\Omega_{\textbf{y}} \times \Omega_x$, where $\Omega_{\textbf{y}} := \{(z,\textbf{y}_{-i}):  \mathcal{R}_i^\top z_i +\mathcal{S}_i^\top \textbf{y}_{-i}^i \in \interior \Omega, \forall i\}$. Additionally, under Assumption \ref{asmp:nonlinearPlant}, $h(x):=\col(G^\top \nabla V(x),\mathcal{S}\textbf{1}_N\otimes (G^\top \nabla V(x)))$, $\frac{\partial f_1}{\partial z}$ and $\frac{\partial f_1}{\partial \textbf{y}_{-i}}$ have bounded first partial derivatives on any compact subset of $\Omega_{\textbf{y}} \times \Omega_x$. Finally,\vspace{-0.2cm}
\begin{align*}
&\frac{\partial f_2(y,\textbf{1}_N\!\otimes\! y,x)}{\partial x}\!=\! \frac{\partial f}{\partial x}\!-\!kG\Big[\frac{\partial F}{\partial y}G^\top \nabla^2 V\!+\! \frac{\partial^2 \phi}{\partial y^2}G^\top \nabla^2 V \Big]
\end{align*}
Under Assumption \ref{asmp:C2}, $\frac{\partial f_2(y,\textbf{1}\otimes y,x)}{\partial x}$ is Lipschitz continuous on any compact subset of $\Omega_x$.

Therefore, by Theorem 11.2 in \cite{hk02}, for each $k\geq k^*$ and for each compact subset $\bar \Omega \subset \Omega_x$, there exists $\bar\epsilon>0$ such that for all $x(0) \in \bar \Omega$ and for all $0<\epsilon<\bar\epsilon$ \eqref{eq:CentreDynamics2} has a unique solution $x(t)$ and $(z(t),\textbf{y}_{-i}(t))$ and for some $\kappa_1>0$\vspace{-0.2cm}
\begin{align*}
\|x(t)-\hat x(t)\|\leq \epsilon \kappa_1
%\|\tilde{\textbf{z}}(t)-G^\top \nabla V(\hat x(t))-\hat{\tilde{\textbf{y}}}_{-i}(t/\epsilon)\|\leq \epsilon \kappa_2
\end{align*}
%where $\hat{\textbf{y}}_{-i}(t)$ and $\hat x(t))$ are the solutions of the boundary-layer and reduced problems. 
where $\hat x(t)$ is the solution to the reduced system. By Lemma \ref{lemma:FI}, $\Omega_x$ is positively invariant for the reduced system. Consider $\kappa_x:=\inf_{t\geq0}\|\hat x(t)\|_{\partial \Omega_x}$ and let $\epsilon^{*}=\min\{\bar\epsilon,\frac{\kappa_x}{\kappa_1}\}$. Then for all $0<\epsilon<\epsilon^{*}$ we have that $\|x(t)-\hat x(t)\|<\kappa_x$ and thus $x(t)\in \Omega_x$ and $y(t) \in \interior \Omega$ for all $t\geq 0$.
\end{proof}

\begin{thm}
Under Assumptions \ref{asmp:cost} and \ref{asmp:pseudo}(b) through \ref{asmp:C2}, there exists $k^*>0$ and $\epsilon^{**}>0$ such that for all $k\geq k^*$ and $0<\epsilon<\epsilon^{**}$, the equilibrium $(\bar z,\bar{\textbf{y}}_{-i},x^*)=(y^*,\mathcal{S}\textbf{1}_N\otimes y^*,\pi(y^*))$ of \eqref{eq:CentreDynamics2}, where $y^*$ is the NE of \eqref{eq:NE}  and $\varepsilon$-vGNE of \eqref{eq:GNE}, is exponentially stable. Moreover, for all $x(0)$ such that $y(0)\in \interior \Omega$, for all $(z(0),\textbf{y}_{-i}(0) \in R_y^{x(0)}$, the constraint $y(t)\in \Omega$ is satisfied for all $t\geq0$.
\end{thm}

\begin{proof}
	Following from the proof of Lemma \ref{lemma:centreInv}, if $k\geq k^*:=\frac{\beta}{\mu}$, then $x=x^*$ is exponentially stable for the reduced system, \eqref{eq:reduced}. $(\tilde z,\tilde{\textbf{y}}_{-i})=(0,0)$ is exponentially stable for the boundary-layer system, \eqref{eq:boundary}. Therefore, by Theorem 11.4 in \cite{hk02}, there exists $\hat\epsilon>0$ such that for all $0<\epsilon<\hat\epsilon$, $(z,\bar{\textbf{y}}_{-i},x^*)=(y^*,\mathcal{S}\textbf{1}_N\otimes y^*,x^*)$ is an exponentially stable equilibrium of \eqref{eq:CentreDynamics2}.
	
	Moreover, by Lemma \ref{lemma:centreInv}, for all $(z(0),\textbf{y}_{-i}(0),x(0))$ such that $y(0)\in \interior \Omega$ and $(z(0),\textbf{y}_i(0)) \in R_y^{x(0)}$, there exists $\epsilon^{*}>0$ such that the constraint $y(t)\in\Omega$ is satisfied for all $t\geq0$. Then take $\epsilon^{**} = \min\{\hat\epsilon,\epsilon^*\}$.
\end{proof}

\section{OSNR Example} \label{sec:OSNR}
Consider an optical-signal-to-noise ratio (OSNR) model for wave division multiplexing links with ten channels. Each channel chooses its transmission power $y_i$ in order to maximize its signal-to-noise ratio. The link is assumed to have a maximum transmission power, $P_0$. This leads to a game given by the following set of optimization problems \vspace{-0.2cm}
\begin{align} \label{eq:osnrGNE}
\begin{split}
\min_{y_i}\quad &a_iy_i-b_i\ln\Big(1+c_i \frac{y_i}{n_i^0+\sum_{ j \neq i }\Gamma_{ij} y_j}\Big)\\
\st\quad  &0\leq y_i\\
&\textstyle{\sum_{j\in\mathcal{I}}} y_j \leq P_0
\end{split}
\end{align}
where $a_i>0$ is a pricing parameter, $b_i>0$,  $\Gamma = [\Gamma]_{ij}$ is the link system matrix and $n_i^0$ is the channel noise power, with parameters as in \cite{sp16}. 
%By using the inexact penalty method, we can convert \eqref{eq:osnrGNE} into the following unconstrained NE seeking problem \vspace{-0.2cm}
%\begin{align*}
%\min_{y_i}\quad &\begin{cases} a_iy_i-b_i\ln\Big(1+c_i \frac{y_i}{n_i^0+\sum_{ j \neq i }\Gamma_{ij} y_j}\Big)-\rho \log(y_i)\\
%\quad -\rho \log(P_0\!-\!\textstyle{\sum_{j\in\mathcal{I}}}y_j),\quad 0\!<\! y_i\!< \!P_0\! -\!\sum_{j\neq i} y_j\\
%+\infty,\quad \text{else} \end{cases}
%\end{align*}
%where $\rho \!>\!0$. 
We consider that each autonomous agent uses a fully-distributed partial-information gradient-play scheme, with agent dynamics given by \eqref{eq:int} and \eqref{eq:CenterFeedback}. In order to get the action information, action estimates are communicated over graph $\mathcal{G}_c$, Fig. \ref{fig:GraphOSNR}. %This can be seen as a partial-information version of the second-order optimization methods of the form \eqref{eq:
%\begin{figure}[h!]
%	\vspace{-0.26cm}
%	\centering
%	\includegraphics[ width =4.5cm]{GraphPlotOSNR}
%	\caption{Communication graph, $G_c$, $\lambda_2 = 2.6158$}
%	\label{fig:GraphOSNR}
%\end{figure}
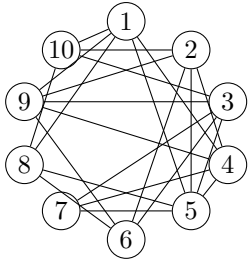
\begin{figure}[h!]
	\vspace{-0.26cm}
	\centering
	\begin{tikzpicture}[auto, node distance=1cm]
	\node [gnode] (p1) {$1$};
	\node [gnode, below left = 0.0175 and 0.5 cm of p1] (p10) {$10$};
	\node [gnode, below right = 0.0175 and 0.5 cm of p1] (p2) {$2$};
	\node [gnode, below right =  0.325 cm and 0.125 cm of p2] (p3) {$3$};
	\node [gnode, below = 0.325 cm of p3] (p4) {$4$};
	\node [gnode, below left = 0.25 cm and 0.125 cm of p4] (p5) {$5$};
	\node [gnode, below left = 0.0175 and 0.5 cm of p5] (p6) {$6$};
	\node [gnode, above left = 0.0175 and 0.5 cm of p6] (p7) {$7$};
	\node [gnode, above left =  0.25 cm and 0.125 cm of p7] (p8) {$8$};
	\node [gnode, above = 0.325 cm of p8] (p9) {$9$};
	
	\draw [-] (p1) -- node []{}(p8);
	\draw [-] (p1) -- node []{}(p5);
	\draw [-] (p1) -- node []{}(p4);
	\draw [-] (p1) -- node []{}(p9);
	\draw [-] (p1) -- node []{}(p10);
	
	\draw [-] (p2) -- node []{}(p4);
	\draw [-] (p2) -- node []{}(p5);
	\draw [-] (p2) -- node []{}(p6);
	\draw [-] (p2) -- node []{}(p9);
	\draw [-] (p2) -- node []{}(p10);
	
	\draw [-] (p3) -- node []{}(p5);
	\draw [-] (p3) -- node []{}(p6);
	\draw [-] (p3) -- node []{}(p7);
	\draw [-] (p3) -- node []{}(p9);
	\draw [-] (p3) -- node []{}(p10);
	
	\draw [-] (p4) -- node []{}(p5);
	\draw [-] (p4) -- node []{}(p7);
	\draw [-] (p4) -- node []{}(p9);
	
	\draw [-] (p5) -- node []{}(p7);
	\draw [-] (p5) -- node []{}(p8);
	
	\draw [-] (p6) -- node []{}(p8);
	\draw [-] (p6) -- node []{}(p9);
	
	\draw [-] (p8) -- node []{}(p10);

	\end{tikzpicture}
	\caption{Communication graph, $\mathcal{G}_c$, $\lambda_2 = 2.6158$} 
	\label{fig:GraphOSNR}
\end{figure}

Fig. \ref{fig:OSNR2} shows transmission power, $y_i$, for each agent over time using $v_i=1$, $k_i=\frac{1}{2}$ and $\rho = 0.1$. Each agent always has positive power and the total power usage on the link is less than the maximum, Fig. \ref{fig:TOTPOW2}, meaning the constraints are satisfied for all time.
\begin{figure}[h!]
	\vspace{-0.3cm}	
	\centering
	\includegraphics[trim=0cm 0cm 0cm 0cm,width=2.8in]{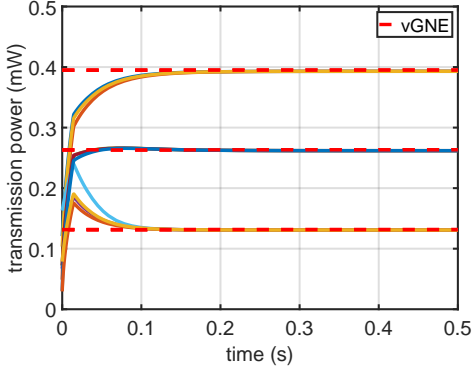}
	\caption{Individual transmission powers, $y_i$}\label{fig:OSNR2}  %partial-information disturbance rejection algorithm 
	\vspace{-0.3cm}	
\end{figure}
\begin{figure}[h!]
	\vspace{-0.3cm}	
	\centering
	\includegraphics[trim=0cm 0cm 0cm 0cm,width=2.8in]{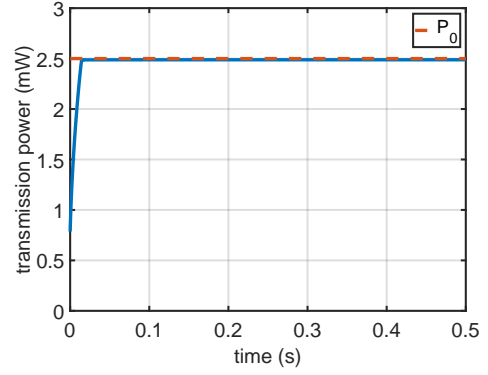}
	\caption{Total power usage, $\sum_{i\in\mathcal{I}} y_i$}\label{fig:TOTPOW2}  %partial-information disturbance rejection algorithm 
	\vspace{-0.3cm}	
\end{figure}
\section{Velocity Synchronization} \label{sec:robots}
Next, we consider a velocity synchronization problem for a group of flexible mobile robots. Consider a group of five flexible mobile robots moving in a line. Each robot is modelled as two masses connected by a non-linear spring, with force $\psi_i(\cdot)$ and damper with coefficient $\gamma_i$. Let $d_i$ denote robot $i$'s position, mass $M_i$, and $\eta_i$ the position of its appendage, mass $m_i$. A force $u_i$ is applied to mass $M_i$. Letting $x_i^1 = d_i-\eta_i$, $x_i^2 = \dot d_i$ and $x_i^3 = \dot \eta_i$, the dynamics of robot $i$ are given by\vspace{-0.2cm}
\begin{align} \label{eq:MobileRobot}
\mathcal{P}_i: \begin{cases}
\dot x_i^1 = x_i^2-x_i^3\\
M_i\dot x_i^2 = -\psi_i(x_i^1)-\eta_i(x_i^2-x_i^3)+u_i\\
m_i\dot x_i^3 = \psi_i(x_i^1)+\eta_i(x_i^2-x_i^3)\\
y_i = x_i^2
\end{cases}
\end{align}
It can be shown that if $\psi_i$ is strongly-monotone and Lipschitz continuous and $\psi_i(0)=0$  , \eqref{eq:MobileRobot} satisfies Assumption \ref{asmp:nonlinearPlant} with $V_i(x_i) = \int_0^{x_i^1} \psi_i(x_i^1) dx_i^1 + \frac{M_i}{2}(x_i^2)^2+ \frac{m_i}{2}(x_i^3)^2$. 

We consider the following leader-follower problems: \vspace{-0.2cm}
\begin{align} \label{eq:sensorNetGNE}
\begin{split}
\min_{y_i}\quad &(y_i-y_{i-1})^2\\
\st\quad  &(y_i-y_{i-1})^2\leq \delta_{i,\text{max}}^2
\end{split}
\end{align}
where $y_0 := v_0$, the reference velocity for the leader. 

We simulate a group of five robots using the full-information feedback \eqref{eq:nonlinFIfeedback} with $v_0 = 3$, $d_{i,\text{max}} = d_{\text{max}} = 3$, $M_i=m_i=1$, $\eta_i=1$, $\psi_i(x_i^1)=-x_i^1-\text{atan}(x_i^1)$ and $\rho = 0.1$. Figure \ref{fig:robots} shows that the robots synchronize to $v_0=3$. Figure \ref{fig:robotdistance} shows that the velocity difference between neighbours never exceeds $d_{\text{max}}$.
\begin{figure}[h!]
	\vspace{-0.3cm}	
	\centering
	\includegraphics[trim=0cm 0cm 0cm 0cm,width=2.8in]{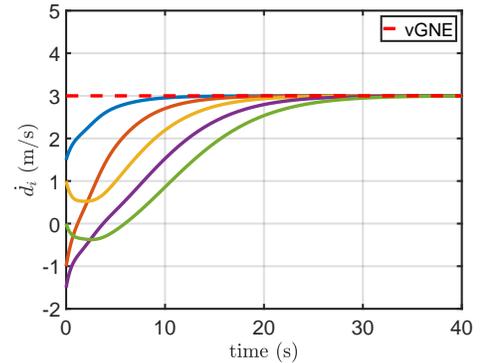}
	\caption{Velocities, $\dot d_i$, of the five robots}\label{fig:robots}  %partial-information disturbance rejection algorithm 
	\vspace{-0.3cm}	
\end{figure}
\begin{figure}[h!]
	\vspace{-0.3cm}	
	\centering
	\includegraphics[trim=0cm 0cm 0cm 0cm,width=2.8in]{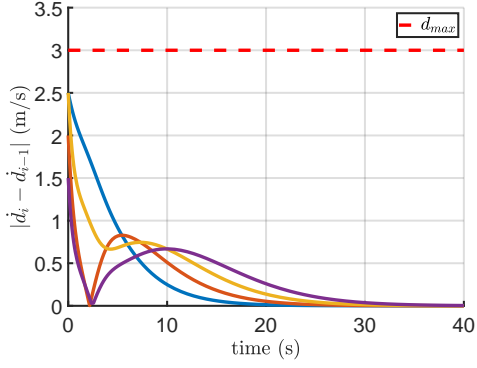}
	\caption{Difference in velocities between neighbours over time}\label{fig:robotdistance}  %partial-information disturbance rejection algorithm 
	\vspace{-0.3cm}	
\end{figure}

\section{Conclusions}
In this paper, we present a novel approach to solving GNE problem using an inexact penalty function on the inequality constraints that converts the GNE seeking problem into an NE seeking problem. We then consider full- and partial-information gradient-play feedbacks for dynamic agents with passive dynamics. We prove both convergence to the NE and constraint satisfaction for all time.

Future work will explore using a time-varying penalty function to guarantee exact convergence to the vGNE. Additionally, we will consider other types of agent dynamics and time-varying cost-functions and constraints.

\appendix

\subsection{Proof of Lemma \ref{lemma:FIEq}} \label{appendix:FIEq}
	By Assumptions \ref{asmp:cost} and \ref{asmp:pseudo}, $y^*$ is the unique point such that $F(y^*)\!+\!\nabla \phi(y^*) \!=\! 0$. Let $x^* \!=\! \pi(y^*)$. By Assumption \ref{asmp:nonlinearPlant} and \eqref{eq:NEcondition},\vspace{-0.2cm}
\begin{align} \label{eq:EQ}
f(\pi(y^*))-G(F(y^*)+\nabla \phi(y^*)) = 0
\end{align}
therefore $x^* = \pi(y^*)$ is an equilibrium point of \eqref{asmp:nonlinearPlant}. Now suppose there is another equilibrium $\bar x$, we have\vspace{-0.2cm}
\begin{align*}
0 &= f(\bar x)-G(F(\bar y) + \nabla \phi(\bar y))\\
\bar y &= G^\top \nabla V(\bar x)
\end{align*}
Let $w = \nabla V(\bar x)$ by strong-monotonicity of Lipschitz continuity of $\nabla V$, $\nabla V$ is invertible. Thus, we have \vspace{-0.2cm}
\begin{align*}
f(\nabla V^{-1}(w)) = G (F(G^\top w) + \nabla \phi(G^\top w))
\end{align*}
Let $w^* = \nabla V(x^*)$. From \eqref{eq:EQ}, we get \vspace{-0.2cm}
\begin{align*}
&f(\nabla V^{-1}(w))-f(\nabla V^{-1}(w^*))\\
&\quad= G [F(G^\top w)\! +\! \nabla \phi(G^\top w)\!-\!F(G^\top w^*)\!-\!\nabla \phi(G^\top w^*)]
\end{align*}
Left multiplying by $(w-w^*)^\top$ gives \vspace{-0.2cm}
\begin{align*}
&(w-w^*)^\top [f(\nabla V^{-1}(w))-f(\nabla V^{-1}(w^*))]= (w\!-\!w^*)^\top G \\
&\quad[F(G^\top w)\! +\! \nabla \phi(G^\top w)\!-\!F(G^\top w^*)\!-\!\nabla \phi(G^\top w^*)]
\end{align*}
By monotonicity of $F+\nabla \phi$ and of $-f\circ \nabla V^{-1}$, we have \vspace{-0.2cm}
\begin{align*} 
&0\geq (w-w^*)^\top f(\nabla V^{-1}(w))-f(\nabla V^{-1}(w^*))\\
&= (w-w^*)^\top G [F(G^\top w)\! +\! \nabla \phi(G^\top w)\\
&\qquad -\!F(G^\top w^*)\!-\!\nabla \phi(G^\top w^*)] \geq 0
\end{align*}
Then, by strict-monotonicity of $F$, $G^\top w = G^\top w^*$, implying $\bar y = y^*$, NE of \eqref{eq:NE} and $\varepsilon$-vGNE of \eqref{eq:GNE} by Lemma \ref{lemma:epsilon}. Thus, by Lemma \ref{lemma:EqUnique}, $\bar x = x^*$. Therefore, $x^* = \pi(y^*)$ is the unique equilibrium of \eqref{eq:FIdynamics}.

\subsection{Proof of Lemma \ref{lemma:practicalSet}} \label{appendix:practicalSet}
	First, we show that $S \subset \Omega_x:=\{x\in \reals^n:g(G^\top \nabla V(x))<0\}$ is compact. Note that $S = \{x \in \reals: V^{x^*}(x) - c \leq 0\} \cap \{x \in \reals: \phi^{x^*}(x) - d \leq 0\}$. Which is closed, since both functions are continuous. Since, the first set is compact, the intersection of the two is also compact.

Next, cf. Definition \ref{def:practicalSet}, Condition 2, we need to show that for all $x \in S$, there exists $z \in \reals^n$ such that $V^{x^*}(x)-c+\nabla V^{x^*}(x)^\top z<0$ and $\phi^{x^*}(x)-d + \nabla \phi^{x^*}(x)^\top z < 0$.
Now, consider $x \in S$. There are four cases to be considered.
\begin{enumerate}
	\item $V^{x^*}(x)-c < 0$, $\phi^{x^*}(x)-d < 0$\\
	In this case, the inequalities hold trivially for $z=0$.
	\item $V^{x^*}(x)-c < 0$, $\phi^{x^*}(x)-d = 0$\\
	Let $z=-G\eta $, where $\eta =( \nabla \phi(G^\top \nabla V(x))-\nabla \phi(G^\top \nabla V(x^*))$. Then \vspace{-0.2cm}
	\begin{align*}
	&V^{x^*}(x)-c+\nabla V^{x^*}(x)^\top z=\\
	&V^{x^*}(x)-c - (\nabla V(x) -\nabla V(x^*))^\top G\eta < 0
	\end{align*}
	by strict-convexity of $\phi$ and the fact that $G^\top \nabla V(x) \neq G^\top \nabla V(x^*)$ for all $x$ such that $\phi^{x^*}(x)=d$. Additionally,\vspace{-0.2cm}
	\begin{align*}
	&\phi^{x^*}(x)-d + \nabla \phi^{x^*}(x)^\top z =-\eta^\top G^\top\nabla^2 V(x)^\top G\eta< 0
	\end{align*}
	by convexity of $\phi(y)$, strong-convexity of $V(x)$, full-column rank of $G$ and since $G^\top \nabla V(x) \neq G^\top \nabla V(x^*)$ when $\phi^{x^*}(x)>0$.
	\item $V^{x^*}(x)-c = 0$, $\phi^{x^*}(G^\top \nabla V(x))-d < 0$\\
	Let $z = -a(\nabla V(x)-\nabla V(x^*))$, where $0<a<\frac{d-\phi^{x^*}(x)}{|\nabla( \phi^{x^*}(x))^\top(\nabla V(x)-\nabla V(x^*))|}$. This gives \vspace{-0.2cm}
	\begin{align*}
	&V^{x^*}(x)-c+\nabla V^{x^*}(x)^\top z =\\
	&-a(\nabla V(x) -\nabla V(x^*))^\top(\nabla V(x)-\nabla V(x^*)) < 0
	\end{align*}
	by strong-convexity of $V(x)$. Furthermore, \vspace{-0.2cm}
	\begin{align*}
	&\phi^{x^*}(x)-d + \nabla \phi^{x^*}(x)^\top z =\\
	&\phi^{x^*}(x)-d - a \nabla \phi^{x^*}(x)^\top (\nabla V(x)-\nabla V(x^*))\\&\leq  \phi^{x^*}(x)-d + a|\nabla \phi^{x^*}(x)^\top (\nabla V(x)-\nabla V(x^*))|<0
	%<\phi^{x^*}(x)-d + d - \phi^{x^*}(x) = 0
	\end{align*}
	\item $V^{x^*}(x)-c = 0$, $\phi^{x^*}(G^\top \nabla V(x))-d = 0$\\
	Let $z$ as in Case 2. Then 
	$V^{x^*}(x)-c+\nabla V^{x^*}(x)^\top z<0$ and
	$\phi^{x^*}(x)-d + \nabla \phi^{x^*}(x)^\top z<0$
	%&-(\nabla V(x) -\nabla V(x^*))^\top G\big [ \nabla \phi(G^\top \nabla V(x))\\&-\nabla \phi(G^\top \nabla V(x^*))\big] < 0
	as in Case 2.
\end{enumerate}
Next, to show condition (3) in Definition \ref{def:practicalSet}, consider the vector field $f_0(x) = (\phi^{x^*}(x)-d)(\nabla V(x) - \nabla V(x^*))-G\eta$, where $\eta = ( \nabla \phi(G^\top \nabla V(x))-\nabla \phi(G^\top \nabla V(x^*))$ and $a \in \reals$, which is Lipschitz continuous on $S$ since $S \subset \Omega_x$ is compact and $f_0$ is $C^1$ on $\Omega_x$. Here, we need to check three cases
\begin{enumerate}
	\item $V^{x^*}(x)-c < 0$, $\phi^{x^*}(x)-d = 0$\vspace{-0.2cm}
	\begin{align} \label{eq:Vec1}
	L_{f_0}\phi^{x^*}= -\eta^\top G^\top \nabla^2 V(x) G\eta< 0
	\end{align}
	by $\phi(x)$ strictly convex, $V(x)$ strongly-convex, rank of $G$ and that $G^\top \nabla V(x) \neq G^\top \nabla V(x^*)$ when $\phi^{x^*}(x)>0$.
	\item $V^{x^*}(x)-c = 0$, $\phi^{x^*}(x)-d < 0$\vspace{-0.2cm}
	\begin{align*}
	L_{f_0}V^{x^*}=&(\phi^{x^*}(x)-d)\|\nabla V(x)-\nabla V(x^*)\|^2\\& -\!(\nabla V(x)\!-\!\nabla V(x^*))^\top G\eta< 0
	\end{align*}
	by strong-convexity of $V(x)$, $\phi^{x^*}(x)<d$ and convexity of $\phi(y)$.
	\item $V^{x^*}(x)-c = 0$, $\phi^{x^*}(x)-d = 0$\\
	We have $L_{f_0}\phi^{x^*}< 0$
	as in \eqref{eq:Vec1}. Furthermore \vspace{-0.2cm}
	\begin{align*}
	L_{f_0}V^{x^*} \!=\!&-\!(\nabla V(x)\!-\!\nabla V(x^*))^\top G\eta< 0
	\end{align*}
	by $\phi(y)$ strictly-convex and $G^\top\nabla V(x) \neq G^\top \nabla V(x^*)$ when $\phi^{x^*}(x)>0$.
\end{enumerate}
Therefore, by Definition \ref{def:practicalSet}, we have that $S$ is a practical set.

\subsection{Proof of Lemma \ref{lemma:PIEq}} \label{appendix:PIEq}
First, we show that $(\mathcal{S}\textbf{1}_N \otimes y^*, \pi(x^*))$ is an equilibrium point of \eqref{eq:FIdynamics}. At $(\mathcal{S}\textbf{1}_N \otimes y^*, \pi(x^*))$, $\textbf{y} = \textbf{1}_N \otimes y^*$. Thus, using the fact that $\textbf{F}(\textbf{1}_N \otimes y) = F(y)$ for any $y$, \eqref{eq:PIDynamics} becomes\vspace{-0.2cm}
\begin{align*}
\dot{\textbf{y}}_{-i} &= -\mathcal{S}\textbf{L}\textbf{1}_N \otimes y^* = 0\\
\dot x &= f(\pi(x^*))\! -\!G(F(y^*)\!+\!\nabla \phi(y^*)\!+\!\mathcal{R}\textbf{L}\textbf{1}_N \otimes y^*)\!=\!0
\end{align*}
Therefore $(\mathcal{S}\textbf{1}_N \otimes y^*, \pi(x^*))$ is an equilibrium point of \eqref{eq:PIDynamics}. Now, suppose there is another equilibrium point $(\bar{\textbf{y}}_{-i},\bar x)$. From \eqref{eq:PIDynamics},\vspace{-0.2cm}
\begin{align} \label{eq:Eq1}
&-\mathcal{S}\textbf{L}(\mathcal{S}^\top \bar{\textbf{y}}_{-i}-\mathcal{R}^\top \bar y) = 0\\
\label{eq:Eq2}
&f(\bar x)\! -\!G(\textbf{F}(\bar{\textbf{y}})\!+\!\nabla \phi(\bar y)\!+\!\mathcal{R}\textbf{L}\bar{\textbf{y}})=0
\end{align}
From \eqref{eq:Eq1}, $\bar{\textbf{y}}_{-i} = \textbf{1}_N\otimes\bar y$. Using  $\textbf{F}(\textbf{1}_N \otimes \bar y) = F(\bar y)$, \eqref{eq:Eq2} becomes
$f(\bar x)\! -\!G(F(\bar{y})+\nabla \phi(\bar y))=0$.
Following the proof of Lemma \ref{lemma:FIEq}, $\bar x = \pi(x^*)$ and $\bar y = y^*$.
\bibliographystyle{styles/IEEETran}
\bibliography{references}

\end{document}